\documentclass[12pt,english]{article}
\usepackage[utf8]{inputenc}
\usepackage[english]{babel}
\usepackage{float}
\usepackage{booktabs,longtable}
\usepackage{amsmath}
\usepackage{amssymb,amsthm,amsfonts}
\usepackage{enumerate}
\usepackage{tikz}
\usetikzlibrary{arrows, arrows.meta, positioning, calc}
\usepackage[authoryear,round]{natbib}
\usepackage[hyphens]{url}
\usepackage{hyperref}
\usepackage[a4paper,body={16cm,23.7cm}]{geometry}

\hypersetup{colorlinks=true, linkcolor=cyan, citecolor=cyan, urlcolor=black, breaklinks=true}
\newtheorem{theorem}{Theorem}
\newtheorem{lemma}{Lemma}

\newtheorem{proposition}{Proposition}

\newtheorem{example}{Example}

\begin{document}

\title{\textbf{The geometric adjudication of water rights in international rivers}
\thanks{We thank William Thomson for helpful comments and suggestions. We gratefully acknowledge grants PID2023-147391NB-I00 and PID2023-146364NB-I00, respectively, funded by MCIU/AEI/10.13039/501100011033 and FSE+.}
}
\author{\textbf{Ricardo Mart\'{\i}nez}\thanks{Universidad de Granada.}$\qquad$ \textbf{Juan D. Moreno-Ternero}\thanks{Universidad Pablo de Olavide.}} 

\maketitle

\begin{abstract}
We study the adjudication of water rights in international rivers. We characterize allocation rules that formalize focal principles to deal with water disputes in a basic model. Central to our analysis is a family of geometric rules that implement concatenated transfers downstream. They can be seen as formalizing Limited Territorial Sovereignty, as suggested in the Rio Declaration on Environment and Development. We apply our rules to the case of the Nile River, with a long history of disputes between downstream and upstream nations.   
\end{abstract}

\medskip{}

\medskip{}  \medskip{}

%\textbf{Statements and declarations} \\ The authors have no conflicts of interest to declare that are relevant to the content of this article. \\ The authors declare that they have no financial relationships that could have influenced the outcomes of this research.

\newpage

%%%%%%%%%%%%%%%%%%%%%%%%%%%%%%%%%%%%%%%%%%%%%%%%%%%%%

\section{Introduction}
Many natural resources are over-exploited nowadays. %The allocation by the state of formal property rights creates a decentralized system that can be used to limit common-pool losses 
A point in case is groundwater, whose depletion is a pressing problem worldwide \citep[e.g.,][]{Jasechko2024}. In contrast to rarely implemented protocols such as rationing \citep[e.g.,][]{Ryan2022}, individual property rights to groundwater are more widespread \citep[e.g.,][]{Edwards2024}. %However, groundwater is expensive to monitor and heterogeneous over space, making it difficult to fully enforce rights when groundwater stocks are diminishing. Groundwater rights often lack full transferability, so the value of the property rights must be inferred using methods such as the hedonic price model.
Water rights have long been a concern for economists \citep[e.g.,][]{ Libecap2011}.\footnote{They were already mentioned in the lead article of the very first issue of the American Economic Review \citep[e.g.,][]{Coman2011}.} 
Under the so-called riparian doctrine, %each property owner fronting on a lake or stream has a right to the unimpaired use of the waterway, regardless of the location of his property along the water-way and regardless of the time at which the property is acquired or use made of the waterway. Consequently, 
rights to water are only usufructuary, although certain diversions of water by riparian rights holders are permissible \citep[e.g.,][]{burness1979appropriative}. %: strictly speaking the right holder may not diminish the flow of water by physically consuming it as this would impair the rights of other riparians. In practice the courts have held that" reasonable" diversions of water by riparian rights holders are permissible, but there are still severe restrictions on such diversions, coupled with uncertainty as to how a court will view any specific diversion.
Under the appropriative doctrine (widely adopted in the US West), the right to a certain amount of water is established and maintained only through use \citep[e.g.,][]{burness1979appropriative}. 
In general, water rights are often ambiguous and costly to enforce, which create frequent disputes and distortions \citep[e.g.,][]{Donna2024}. %, distort resource allocation and production decisions, deter investment, and act as a barrier to transferring water rights 
Historically, judges resolved disputes over water rights on an ad hoc basis. %: when a problem arose, a judge would review the claims of all involved parties, often at significant expense, and determine the rights of each constituent. Over the past thirty years, 
%it has become commonplace for states to minimize disputes by undertaking 
More recently, aiming to minimize disputes, a process of general basin adjudication has been gradually endorsed. With this process, a court formally verifies which users within a watershed have valid rights and what their entitlements are \citep[e.g.,][]{Browne2023}. %However, large-scale adjudications often take decades to complete and involve substantial litigation costs because of their scale and complexity. Initially, legal experts expected large adjudications to be completed within a few years. Yet as litigation caused the process to drag on for decades, the cost to taxpayers became increasingly steep, leading some legal scholars to question the efficacy of general basin adjudication 

The aim of this paper is to provide a basic model for the adjudication of riparian rights. In our model, agents (to be naturally interpreted as countries, although they could also be considered as regions, or towns) are located along a river with a linear structure. For each agent, there is a river inflow from tributaries. And we want to design rules that associate with each possible profile of inflows an allocation indicating the amount of water each agent gets. We require from the outset that allocations are \emph{feasible} and \emph{non-wasteful}. That is, at each river location, the overall amount allocated to upstream agents does not exceed the overall inflow upstream; and the overall amount provided in the allocation coincides with the overall inflow.

A natural first step to design rules is to formalize existing focal principles to prevent or resolve water disputes within an international river basin \citep[e.g.,][]{Kilgour1995}. %As there is no binding international law governing the allocation of water in international rivers, those principles essentially constitute the only available guidelines to obtain a fair allocation of water rights for international rivers \citep[e.g.,][]{Brink2012}. %list several principles on water rights%can be taken into account . Since , these principles, in combination with two legal texts called the 1966 Helsinki Document (see Kilgour and Dinar) and the 1997 UN convention, are the only guidelines that are available in determining a Ã¢??fairÃ¢?? welfare distribution and thus a Ã¢??fairÃ¢?? solution for the class of cooperative TU-games on river systems with multiple springs. 
One is the so-called principle of Absolute Territorial Sovereignty (ATS), also known as the Harmon doctrine. This principle asserts that countries can use any water that flows into the river on their territory, without taking into account the downstream consequences. In our setting, this will be formalized as the so-called \textit{no-transfer} rule, which simply suggests the initial profile of inflows as the solution of the distribution problem. %(i.e., the identity rule). 
A somewhat polar principle is the so-called principle of Unlimited Territorial Integrity (UTI). This principle stipulates that a country
may not alter the natural flow of waters passing through its territory in any
manner which will affect the water in another state. %states that the inflow at the territory of some country can be claimed by this country and all its downstream countries. This doctrine , be it upstream or downstream.
In our setting, this will be formalized as the so-called \textit{full-transfer} rule, which transfers each inflow to the most downstream agent. 

A compromise between the above two extreme principles, dubbed Limited Territorial Sovereignty, was suggested in the Rio Declaration on Environment and Development, 1992, whose Principle 2 is written as follows:\footnote{Rio Declaration on Environment and Development, United Nations Conference on Environment and Development, 13 June 1992, U.N. Doc. A/CONF.151/26 (Vol. I), reprinted in 31 I.L.M. 874 (1992).} %Last accessed, November 25, 2025.}

\textit{``States have, in accordance with the Charter of the United Nations and the
principles of international law, the sovereign right to exploit their own
resources pursuant to their own environmental and developmental policies, and
the responsibility to ensure that activities within their jurisdiction or
control do not cause damage to the environment of other States or of areas
beyond the limits of national jurisdiction."}

%which states that the water belongs to all countries together, instead of considering any country the owner of the water. In our setting this has several possible interpretations. One is that the inflow at some country should be shared equally by this country and all its downstream countries.\footnote{Note that including the own agent in the sharing process is what constitutes the critical difference with respect to the previous rule.} This gives rise to what we call the \textit{Shapley} rule as its expression is akin to the famous value in the literature of cooperative game theory. 
%Thus, the TIBS principle can be seen as a sort of compromise between the previous two principles (ATS and UTI). 

We formalize this principle in our setting by means of a family of rules emerging from compromising between the no-transfer rule and the full-transfer rule, via \textit{geometric} downstream transfers.
%The TIBS principle is also known as the principle of community of interests in the waters or the principle of common management and can, alternatively, be described as follows: Ã¢??the water belongs to all basin states combined, no matter where it enters the river, and each state is entitled to a reasonable and equitable share in the optimal use of the waterÃ¢??, see e.g. Lipper [24] or McCaffrey [26]. According to this description, the principle requires that (i) the water is assigned in such a way that the total welfare of all countries is maximized (optimal use) and (ii) each country gets a (reasonable and equitable) share in the total welfare resulting from an optimal assignment. In this paper we apply the TIBS principle in the following way to a river basin with multiple springs and satiable agents. Suppose that, for one reason or another, the agents along a river with multiple springs are cooperating in two separate coalitions as follows. For some agent, say i, one of the coalitions consists of agent i and all its upstream agents, and the other coalition is its complement, i.e., consists of all other agents. This happens, for instance, when agent i is not willing to cooperate with its unique downstream neighbor. The question that then can be asked is: how should the gain in total welfare that is created when the two coalitions join together into one coalition of all agents, be divided among the agents? E
Precisely, let $\gamma\in [0,1]$. Assume agent 1 retains the fixed portion $\gamma e_1$ of its inflow  and transfers the remainder to agent 2 (the next agent downstream). Then, the disposable inflow for agent 2 consists of its inflow, $e_2$, and the inflow received from agent 1, $(1-\gamma)e_1$. Assume now that agent 2 also retains a $\gamma$ share of its disposable inflow, i.e., $\gamma (e_2 + (1-\gamma)e_1)$, transferring the rest to the next agent downstream. The process continues sequentially until the most downstream agent is reached. This gives rise to a family of (single-parameter) \textit{geometric} rules. When $\gamma=0$, we obtain the full-transfer rule formalizing the UTI principle described above. When $\gamma=1$, we obtain the no-transfer rule formalizing the ATS principle described above. As $\gamma$ varies in $(0,1)$ we have a variety of options formalizing compromises between both principles, to be interpreted as various formalizations of Limited Territorial Sovereignty.

We characterize the family of (single-parameter) geometric rules defined above upon combining four axioms: \textit{scale invariance}, \textit{upstream invariance}, \textit{partial-implementation invariance} and \textit{equal sources}. The first axiom (scale invariance) is standard in the axiomatic literature and it simply states that if all inflows are multiplied by the same factor, then so is the allocation. The second axiom (upstream invariance) states that a shock in a given river location (in the form of changing the inflow therein) does not have an effect upstream (as water flows downstream). The third axiom (partial-implementation invariance) formalizes an invariance property (for downstream agents), when an agent's inflow is increased with the remaining flow from upstream that was not allocated to upstream agents (whose inflows are now considered null). %motivated as follows. Imagine a rule has been applied to a given problem. Now, for a given river location ($i$), assume all of its agents upstream have received the amount the rule specified for them. This generates a residual problem in which the inflows of those upstream agents are set to zero, whereas the inflows for agents located downstream remain unchanged. As for agent $i$ itself, we assume it is endowed with a new inflow resulting from the sum of its old inflow and any remaining flow from upstream that was not allocated (to upstream agents). The invariance requirement for the rule is that, in the residual problem, it assigns the same amount to each of the downstream agents as it did initially. 
Finally, the fourth axiom (equal sources) states that if two rivers have sources with the same (strictly positive) inflow, then these sources receive the same amount.\footnote{By source of a river we understand its most upstream location with a strictly positive inflow.} 
%\newpage

The normative appeal of the latter axiom described above might be questionable, as it ignores the number of downstream agents and their inflows. It turns out that if we drop it from the list of axioms we characterize a larger family of rules, dubbed \textit{multi-parameter geometric rules}, which extend single-parameter geometric rules naturally to allow for individual-specific portions at each stage. That is, agent 1 retains a portion of its inflow $\gamma_1 e_1$ and transfers the remainder to agent 2. Then, the disposable inflow for agent 2 consists of its inflow, $e_2$, and the inflow received from agent 1, $(1-\gamma_1)e_1$. Then, agent 2 retains a portion $\gamma_2$ of its disposable inflow, i.e., $\gamma_2 (e_2 + (1-\gamma_1)e_1)$, transferring the rest to the next agent downstream. The process continues sequentially until the most downstream agent is reached. An interesting feature of this family is that the agent-specific portions may capture some of the characteristics (e.g., population, historical needs/use, etc.) that are not explicitly included in our stylized model, and might nevertheless play a role in the allocation process. 

%e_n + \sum_{k=1}^{n-1} \left( e_k-R^\alpha_k(e) \right)  & \text{if } i=n 
%\end{cases}
%$$

An intriguing multi-parameter geometric rule, outside the family of single-parameter geometric rules, is the so-called \textit{serial rule}, introduced by \cite{Mart2025}, and reminiscent of the namesake cost-sharing rule introduced by \cite{Moulin1992}. 
It splits each agent's inflow equally among all agents located downstream, including the agent itself. That is, the rule balances the legitimate claim each agent has to consume their own inflow with the rights of downstream agents to enjoy part of the resources produced upstream. We show that this rule is the only multi-parameter geometric rule that satisfies the axiom of \textit{neutrality}, which states that in a canonical situation (in which only the source of the river yields positive inflow), the source is awarded an amount equal to the average amount downstream agents get.\footnote{Note that this is weaker than requiring equal division downstream.} In other words, the serial rule is characterized by the combination of \textit{scale invariance}, \textit{upstream invariance}, \textit{partial-implementation invariance} and \textit{neutrality}.

To conclude our analysis, we apply our results to the Nile River, probably one of the most famous rivers worldwide, for which disputes have long existed. %Over 250 million people in the Nile region rely on the river for water supply, and around 10\% already face water scarcity. 
%Water rights along the Nile have been in dispute for a long time. In spite of the so-called Nile River Agreements that regulate water rights among countries, tensions persist, with downstream nations (such as Egypt and Sudan) demanding their share, while upstream states (such as Tanzania, and South Sudan) arguing that these agreements hinder their development. 
Using data from AQUASTAT, we compute the inflows of Tanzania, Uganda, Ethiopia, South Sudan, Sudan, and Egypt, as well as the amount of water they are actually extracting from the river. We acknowledge that the issue of allocating riparian water rights in the Nile River is more complex than what our model captures, given the historical, political, cultural, economic, as well as social ramifications. That is why we resort to our family of multi-parameter geometric rules, which allow to account some of these aspects. More precisely, we look for country-specific portions of inflows, i.e., the vector $\gamma=(\gamma^1,\gamma^2\dots\gamma^n)$, that rationalize the current allocation of water emanating from the Nile River Agreements. We also provide alternative allocations that could be considered, based on the remaining rules we present in our analysis. %We nevertheless believe that our empirical illustration, limited as might be by the theoretical model we consider (thus not encompassing all relevant factors at stake), still oï¬€ers valuable insights.
%Significant disparities exist between the water rights that would result from the application of the rules we introduce (and characterize) in this paper and the actual allocation of water that occurs nowadays. 
%Our analysis implies that the actual allocation of water,, can be rationalized by the %$\lambda$-egalitarian transfer rules, provided we consider very low values of the parameter describing the family. %$\lambda$. Similarly, %by comparing the disparities between the two families of egalitarian (partial or not) transfer rules with the actual allocation, we conclude that the spirit of the Nile River Agreements is more aligned with the axiom of equal sources (which characterizes the egalitarian transfer rules) than with the axiom of upstream total inflow (which characterizes the egalitarian partial-transfer rules).

The rest of the paper is organized as follows. In Section 2, we introduce the model, as well as the main axioms and rules of our analysis. In Section 3, we present our main characterizations, as well as extra characterizations exploring further axioms. In Section 4, we present the illustration to the case of the Nile River. We conclude in Section 5. For a smooth passage, all proofs have been relegated to an Appendix.

\section{The model}\label{model}

%We consider the same model introduced by \cite{Mart2025}. That is, 
A set of \textbf{agents} $N=\{1, \ldots, n\}$ (with $n\ge 3$) are located along a river, which has a linear structure. Lower numbers represent more upstream locations, so that agent 1 is the most upstream agent, agent $n$ is the most downstream agent, and $i\le j$ means that agent $i$ is upstream of agent $j$. %For each $i \in N$, we denote by $U(i)=\{j \in N: j\le i \}$ and $D(i)=\{j \in N: i\le j \}$ the set of upstream and downstream agents of $i$, respectively. 
For each $i \in N$, there is an inflow of the river, denoted by $e_i \geq 0$. Let $e=\left(e_1, \ldots, e_n\right) \in \mathbb{R}^n_+$ be the \textbf{profile of inflows}. Let $\mathcal{D}$ denote the domain of all profiles of inflows. 

Our aim is to provide rules that associate with each profile of inflows an \textbf{allocation} of riparian water rights. That is, another profile indicating the amount of water over which each agent gets a right. As water flows downstream, we require that allocations are \emph{feasible} and \emph{non-wasteful}, i.e.,  $x=\left(x_1, \ldots, x_n\right) \in \mathbb{R}^n_+$ is such that $\sum_{i=1}^k x_i\le\sum_{i=1}^k e_i$, for each $k=1,\dots, n-1$, and $\sum_{i=1}^n x_i=\sum_{i=1}^n e_i$. A \textbf{rule} $R:\mathcal{D}\to \mathbb{R}^n_+$ is a mapping that associates to each profile of inflows $e \in \mathcal{D}$ an allocation $R(e)\in \mathbb{R}^n_+$. 

\subsection{Rules}
The following pair of rules formalize polar ideas regarding the allocation of water rights. The first rule says that each agent retains its inflow. The second rule states that each agent, except for the last one, transfers its whole inflow. Formally, 

\textbf{No-transfer rule}. For each $e \in \mathcal{D}$, and each $i \in N$,
$$
R^{NT}_i(e) = e_i.
$$

\textbf{Full-transfer rule}. For each $e \in \mathcal{D}$, and each $i \in N$,
$$
R^{FT}_i(e) = 
\begin{cases}
    0 & \text{if } i \neq n \\[0.1cm]
    \sum_{k=1}^n e_k & \text{if } i=n 
\end{cases}
$$

The next family of rules represents a compromise between the two extreme options presented above. Let $\gamma\in [0,1]$. Assume agent 1 retains the fixed portion of its inflow $\gamma e_1$ and transfers the remainder to agent 2 (the next agent downstream). Then, the disposable inflow for agent 2 consists of its inflow, $e_2$, and the inflow received from agent 1, $(1-\gamma)e_1$. Assume now that agent 2 also retains a $\gamma$ share of its disposable inflow, i.e., $\gamma (e_2 + (1-\gamma)e_1)$, transferring the rest to the next agent downstream. The process continues sequentially until the most downstream agent is reached. Formally, 

\textbf{Geometric rule of parameter} $\mathbf{\gamma\in[0,1],\, R^\gamma}$: For each $e \in \mathcal{D}$ and each $i \in N$,
$$
R^\gamma_i(e) = 
\begin{cases}
    \gamma e_1 & \text{if } i=1 \\[0.1cm]
    \gamma \left( e_i + \sum_{k=1}^{i-1} \left( e_k-R^{\gamma}_k(e) \right) \right)  & \text{if } i \in \{2,\ldots,n-1\} \\[0.1cm]
    e_n + \sum_{k=1}^{n-1} \left( e_k-R^{\gamma}_k(e) \right)  & \text{if } i=n. 
\end{cases}
$$
Alternatively, 
$$
R^\gamma_i(e) = \gamma \left( e_i + \sum_{k=1}^{i-1} (1-\gamma)^{i-k}e_k \right),
$$
for each $i \in \{1,\dots,n-1\}$, and
$$
R^\gamma_n(e) = e_n + \sum_{k=1}^{n-1} (1-\gamma)^{n-k}e_k.
$$

Notice that, when $\gamma=0$ we obtain the full-transfer rule, whereas when $\gamma=1$ we obtain the no-transfer rule, i.e., $R^0\equiv R^{FT}$, and $R^1\equiv R^{NT}$.

\begin{example} Consider two vectors of inflows: $e^1=(0,36,0,0)$ and $e^2=(12,4,0,10)$. The next table shows how two particular single-parameter geometric rules apply to these problems.
\begin{table}[H]
    \centering
    \begin{tabular}{lcccc}
    \toprule
    & \multicolumn{2}{c}{$\gamma=\frac{1}{2}$} &  \multicolumn{2}{c}{$\gamma=\frac{2}{3}$} \\
    \cmidrule(lr){2-3} \cmidrule(lr){4-5}
    Agents & $R^\gamma(e^1)$ & $R^\gamma(e^2)$ & $R^\gamma(e^1)$ & $R^\gamma(e^2)$ \\ 
    \midrule 
    1 & 0  & 6  & 0  & 8 \\
    2 & 18 & 5  & 24 & $\frac{16}{3}$ \\
    3 & 9  & $\frac{5}{2}$  & 8 & $\frac{16}{9}$ \\
    4 & 9  & $\frac{25}{2}$ & 2 & $\frac{98}{9}$ \\
    \bottomrule
    \end{tabular} 
\end{table}
\end{example}

The family of single-parameter geometric rules can naturally be extended to allow for individual-specific portions at each stage. That is, 

\textbf{Geometric rules of parameter} $\mathbf{\alpha=(\alpha_1,\ldots,\alpha_{n-1},\alpha_n) \in [0,1]^{n-1} \times \{1\},\, R^\alpha}$: for each $e \in \mathcal{D}$ and each $i \in N$,
$$
R^\alpha_i(e) = \alpha_i \left[ e_i + \sum_{k=1}^{i-1} \left( e_k-R^{\alpha}_k(e) \right) \right].
$$

Alternatively, %we can also provide a non-recursive definition of the multi-parameter geometric rules. For each $e \in \mathcal{D}$ and each $i \in N$,
$$
R^\alpha_i(e) = \alpha_i \left[ e_i + \sum_{k=1}^{i-1} \prod_{j=k}^{i-1} (1-\alpha_j)e_k \right].
$$

%Another generalization of the geometric rules arises when the amount each agent keeps along the process is a general function (not necessarily linear) of its inflow. That is, 

%\textbf{Super geometric rules}. For each $G:\mathbb{R}_+ \longrightarrow \mathbb{R}_+$ such that $\alpha(x) \le x$, each $e \in \mathcal{D}$ and each $i \in N$.
%$$
%R^\alpha_i(e) = 
%\begin{cases}
%    \alpha(e_1) & \text{if } i=1 \\[0.1cm]
%    G\left( e_i + \sum_{k=1}^{i-1} \left( e_k-R^\alpha_k(e) \right) \right)  & \text{if } i \in \{2,\ldots,n-1\} \\[0.1cm]
%e_n + \sum_{k=1}^{n-1} \left( e_k-R^\alpha_k(e) \right)  & \text{if } i=n 
%\end{cases}
%$$

An interesting multi-parameter geometric rule, outside the family of single-parameter geometric rules, is the %so-called \textit{serial rule}, introduced by \cite{Mart2025}. In pursuit of rules that balance the legitimate claim each agent has to consume their own inflow with the rights of downstream agents to enjoy part of the resources produced upstream, the serial rule proposes an equitable solution. 
one that splits each agent's inflow equally among all agents located downstream, including the agent itself.%\footnote{The rule is reminiscent of the so-called serial cost-sharing rule introduced by \cite{Moulin1992}.}

\textbf{Serial rule}, $\mathbf{R^{S}}$. For each $e \in \mathcal{D}$ and each $i\in N$,%= 1,2,3,\dots n-1$,
$$ 
R^{S}_i(e) = \frac{e_1}{n}+\frac{e_2}{n-1}+\frac{e_3}{n-2}+\ldots+\frac{e_i}{n-i+1} = \sum_{j \leq i} \frac{e_j}{n-j+1}. 
$$
\subsection{Axioms}
Instead of endorsing directly any particular rule among those introduced above, we resort to the axiomatic approach to select among them. %That is, we consider axioms that formalize several properties with ethical and/or operational content and examine the implications of these axioms, ultimately aiming to characterize specific rules by combining them in various ways.

We first present four axioms already considered by \cite{Mart2025}. %start with a basic requirement stating that if all inflows are multiplied by a same positive real number, then the allocation is too. 

The first axiom states that if all inflows are multiplied by the same factor, then so is the allocation. 

\textbf{Scale invariance}. For each $e \in \mathcal{D}$ and each $\rho \in \mathbb{R}_+$, $R(\rho e)=\rho R(e)$.

The second axiom says that, if the inflow of one agent changes, the allocation of those located upstream remains unaltered. Formally, 

\textbf{Upstream invariance}. For each pair $e,e' \in \mathcal{D}$, such that %$e_i<e'_i$ for some $i \in N$, and 
$e_j = e'_j$ for all $j \in N \backslash \{i\}$, then 
for each $k<i$,
$$
R_k(e) = R_k(e').
$$

%Now, we introduce an axiom referring to the source of a river, to be understood as its most upstream location with a positive inflow. More precisely, 
For each $e \in \mathcal{D}$, its \textit{source} is $s(e)=\min\{k\in \{1,2,\dots n-1\}: e_k>0\}$. The next axiom states that if the sources of two rivers have the same inflow, then these sources receive the same amount.  

\textbf{Equal sources}. For each pair $e,e' \in \mathcal{D}$, such that $e_{s(e)}=e'_{s(e')}$,
$$
R_{s(e)}(e)  = R_{s(e')}(e').
$$

The fourth axiom focuses on the special cases in which the source is the only agent with strictly positive inflow, requiring the rule be \textit{neutral} between the source and the average agent downstream. More precisely, the axiom states that, for those cases, the source should get exactly the average amount downstream agents get. 

\textbf{Neutrality}. For each $e \in \mathcal{D}$, such that $e_i>0$ for some $i \in \{1,\dots, n-1\}$, and $e_j=0$ for all $j \in N \backslash \{i\}$,
$$
R_i(e)=\frac{1}{n-i}\sum_{k>i} R_k(e).
$$

We also introduce a new axiom for this setting, which is an adaptation to this context of an invariance property recently introduced by \cite{Dietzenbacher24} for claims problems. To motivate it, imagine a rule has been applied to a given problem. Now, for each agent $i \in N$, %we define the reduced problem as follows: 
assume all of its upstream agents are awarded as the rule indicates. This generates a residual problem in which the inflows of those upstream agents are set to zero, whereas agent $i$ is endowed with a new inflow resulting from the sum of its old inflow and any remaining flow from upstream that has not been allocated. The inflows for agents located downstream remain as in the original problem. The invariance requirement for the rule is that it assigns the same amount to each of the downstream agents as it did initially. Formally,

\textbf{Partial-implementation invariance}. For each $e \in \mathcal{D}$ and each $i \in N$,
$$
R_{D(i)} = R_{D(i)}\left( 0_{U(i) \backslash \{i\}} , e_i+ \sum_{k=1}^{i-1} (e_k-R_k(e)), e_{D(i) \backslash \{i\}} \right)
$$
%\newpage

As shown in the next section, this new axiom (in this setting) will be instrumental to provide characterization results. %it turns out that the combination of this unique new axiom in this setting, with various of the remaining axioms (which had already been considered in this setting), produces new characterization results.

\section{The results}
\subsection{The main characterizations}
Our first result states that the combination of partial-implementation invariance, upstream invariance, and scale invariance characterizes the multi-parameter geometric rules (the broadest family introduced above).

\begin{theorem} \label{thm_generalized_geom}
A rule satisfies partial-implementation invariance, upstream invariance, and scale invariance if and only if it is a multi-parameter geometric rule.
\end{theorem}

If we add the axiom of equal sources, the family shrinks to the sub-family of single-parameter geometric rules, as stated in the next result.

\begin{theorem}\label{thm_geom}
A rule satisfies partial-implementation invariance, upstream invariance, scale invariance and equal sources if and only if it is a single-parameter geometric rule.
\end{theorem}

Finally, adding neutrality instead of equal sources %to upstream invariance and partial-implementation invariance 
we characterize the serial rule. Actually, scale invariance can be dismissed from the statement (although the serial rule does satisfy that axiom too).

\begin{theorem} \label{thm_shapley}
A rule satisfies partial-implementation invariance, upstream invariance, and neutrality if and only if it is the serial rule.
\end{theorem}

%\section{Further results}

\subsection{Downstream impartiality}

\cite{Mart2025} provide counterpart characterizations to those in Theorems \ref{thm_generalized_geom}-\ref{thm_shapley} replacing partial implementation invariance by an alternative axiom, dubbed \emph{downstream impartiality}. This axiom, akin to upstream invariance, refers to the impact an increase in the inflow of one agent has on the allocation of the others. However, it focuses on the consequences for agents located downstream of the increase. Specifically, this axiom states that if two downstream agents have equal inflows, then the impact should be the same for both of them. Formally,

\textbf{Downstream impartiality}. Let $e,e' \in \mathcal{D}$ be such that $e_i<e'_i$ for some $i \in N$ and $e'_j = e_j$ for all $j \in N \backslash \{i\}$. Then, for each pair $k,l>i$ such that $e_k=e_l$,
$$
R_k(e') - R_k(e) = R_l(e') - R_l(e).
$$

For ease of exposition, we reproduce here the characterization results in \cite{Mart2025}. 

\begin{proposition}\label{thm_add} The following statements hold:
\begin{itemize}
\item A rule satisfies downstream impartiality, upstream invariance and scale invariance if and only if there exists $\delta=(\delta_1, \ldots, \delta_n) \in [0,1]^{n-1} \times \{1\}$ such that, for each $e \in \mathcal{D}$ and $i \in N$,
$$
R^\delta_i(e) = \delta_i e_i + \sum_{k<i} \frac{(1-\delta_k)e_k}{n-k}.
$$
\item A rule satisfies downstream impartiality, upstream invariance, scale invariance and equal sources if and only if there exists $\lambda\in[0,1]$ such that, for each $e \in \mathcal{D}$, and each $i= 1,2,3,\dots n-1$,
$$
R_i^{\lambda}(e) =  %\frac{e_1}{n-1}+\frac{e_2}{n-2}+\frac{e_3}{n-3}\dots \frac{e_{i-1}}{n-i+1}=
\lambda e_i+ (1-\lambda) \sum_{j<i} \frac{e_j}{n-j},
$$
and
$$
R^{\lambda}_n(e) =e_n + (1-\lambda)\sum_{j<n} \frac{e_j}{n-j}.
$$
\item A rule satisfies downstream impartiality, upstream invariance, and neutrality if and only if it is the serial rule.
\end{itemize}
\end{proposition}

At first glance, the structures of the families characterized in Theorem \ref{thm_generalized_geom} and Proposition \ref{thm_add} (first item) are very different. They have a multiplicative and additive form, respectively. However, as stated in both theorems, both satisfy the axioms of scale invariance and upstream invariance. They only differ in one axiom (partial-implementation invariance and downstream impartiality, respectively). A natural question arising from here is whether those axioms are mutually exclusive. The following result demonstrates that they are not. 

Theorem \ref{thm_add_geom} actually characterizes the rules that satisfy those two axioms (partial implementation invariance and downstream impartiality), together with the remaining two generic axioms (scale invariance and upstream invariance). The resulting family is particularly intriguing, as its construction is derived from the no-transfer and serial rules, albeit not through the conventional combinations (convex, midpoint, or weighted) one might expect. 

Formally, let $\beta=(\beta_1,\ldots,\beta_{n-1}) \in [0,1)^{n-1}$. For each $k=1,\dots n-1$, we define the rule $R^{\beta_k}$ as the following single-parameter geometric $R^{\alpha(\beta_k)}$. 

For each $e \in \mathcal{D}$ and each $i \in N$,
$$
R^{\beta_k}_i(e)=R_i^{\alpha(\beta_k)}(e),
$$
where 
$$
\alpha_i(\beta_k) =
\begin{cases}
1 & \text{if } i<k \\
\beta_k & \text{if } i=k \\
\dfrac{1}{n-i+1} & \text{if } k<i<n-1.
\end{cases}
$$

The $\beta$-family, $R^\beta$, is made of all the $R^{\beta_k}$ rules, as %described encompasses the union of all subfamilies $R^{\beta_k}$. The structure of the parameters of the multi-parameter geometric rules that give rise to the $\beta$-family is 
described above. %Each row corresponds to a specific subfamily $R^{\beta_k}$. 
As we can see in Table \ref{table_beta}, these parameters exhibit a particular pattern, integrating both the no-transfer rule ($\alpha_i(\beta_k)=1$) and the serial rule ($\alpha_i(\beta_k)=\frac{1}{n-i+1}$). In each case, $\alpha_k=\beta_k$ is free, whereas for those upstream it is as in the no-transfer rule and for those downstream it is as in the serial rule.

\begin{table}[H]
\begin{center}
\begin{tabular}{cccccccc}
\toprule
Rule & $\alpha_1$ & $\alpha_2$ & $\alpha_3$ & $\alpha_4$ & $\cdots$ & $\alpha_{n-1}$ & $\alpha_n$ \\
\midrule
$R^{\beta_1}$ &$\beta_1$ & $\frac{1}{n-1}$ & $\frac{1}{n-2}$ & $\frac{1}{n-3}$ & $\cdots$ & $\frac{1}{2}$ & 1 \\
$R^{\beta_2}$ &$1$ & $\beta_2$ & $\frac{1}{n-2}$ & $\frac{1}{n-3}$ & $\cdots$ & $\frac{1}{2}$ & 1 \\
$R^{\beta_3}$ &$1$ & $1$ & $\beta_3$ & $\frac{1}{n-3}$ & $\cdots$ & $\frac{1}{2}$ & 1 \\
$R^{\beta_4}$ &$1$ & $1$ & $1$ & $\beta_4$ & $\cdots$ & $\frac{1}{2}$ & 1 \\
$\vdots$ &$\vdots$ & $\vdots$ & $\vdots$ & $\vdots$ & $\vdots$ & $\vdots$ & $\vdots$ \\
$R^{\beta_{n-1}}$ &1 & 1 & 1 & 1 & 1 & $\beta_{n-1}$ & 1 \\
$R^{NT}$ &1 & 1 & 1 & 1 & 1 & 1 & 1 \\
\bottomrule
\end{tabular}
\end{center}
\caption{Structure of the parameters in the $\beta$-family.\label{table_beta}}
\end{table}

\begin{theorem}\label{thm_add_geom}
A rule satisfies partial implementation invariance, downstream impartiality, upstream invariance and scale invariance if and only if it belongs to the $\beta$-family.
\end{theorem}

\section{Application. The Nile River.}

The Nile River, the world's longest river at about 6,650 km, flows north from south of the Equator through northeastern Africa into the Mediterranean Sea. Its basin spans 11 countries, including Tanzania, Uganda, South Sudan, Ethiopia, Sudan, and Egypt. Around 250 million people depend on it for water, and approximately 10\% of them already face water scarcity. %The UN projects that by 2040, this scarcity will affect 35\% of the population in the Nile basin (about 80 million people).
It is therefore not surprising that disputes over Nile water have persisted for decades. Colonial-era treaties still shape its allocation. %The two more remarkable treaties on allocation of riparian water rights are the 1929 Exchange of Notes between the UK and Egypt on irrigation rights and the 1959 Agreement between Sudan and Egypt for full utilization of the river.\footnote{\url{http://www.internationalwatersgovernance.com/nile-river-basin-initiative.html}. Last accessed April 30th, 2025.}
These treaties grant large amounts of water to downstream countries (primarily Egypt and Sudan), mostly ignoring upstream riparian states (such as Kenya, Tanzania and South Sudan). Additionally, they are prohibited from constructing dams or initiating river projects without approval from downstream nations, particularly Egypt. The upstream countries have consistently disputed and questioned these agreements. While Egypt and Sudan demand that their water share must be upheld, the upstream states argue that these agreements are unfair and harm their agricultural and developmental objectives.\footnote{\url{https://arabcenterdc.org/resource/water-conflict-between-egypt-and-ethiopia-a-defining-moment-for-both-countries/}. Last accessed October 30th, 2025.}%The so-called Nile Basin Initiative was launched in February 1999 by the water ministers of the countries that share the river with the aim of ``seeking to develop the river in a cooperative manner, share substantial socioeconomic benefits, and promote regional peace and security" and to ``provide an institutional mechanism, a shared vision, and a set of agreed policy guidelines to provide a basinwide framework for cooperative action."

In this section, we apply the previously introduced (and characterized) rules to provide a framework for cooperative action in the Nile River basin. As mentioned in the introduction, we acknowledge that allocating riparian water rights in the Nile River is a highly complex issue, influenced by historical, political, cultural, economic, and social factors. We nevertheless believe that our empirical illustration, limited as might be by the theoretical model we consider (which, admittedly, does not encompass all relevant factors at stake), still offers valuable insights. 

The Nile River basin is formed by its two main tributaries: the White Nile, originating in the Great Lakes region and flowing through Uganda and South Sudan; and the Blue Nile, which begins at Lake Tana in Ethiopia. These tributaries converge in Khartoum (Sudan), after which the Nile River continues north through Sudan and Egypt before finally emptying into the Mediterranean Sea. Figure \ref{graphNile} provides a visual overview of its course and the major countries within its basin. Figure \ref{fig:NileBasin} reproduces its whole basin.

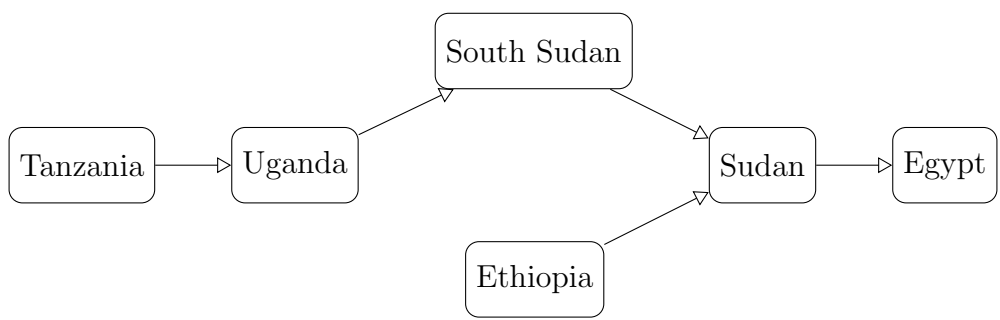
\begin{figure}[h]
\centering
\begin{tikzpicture}[
  node distance=0.5cm and 1cm,
  every node/.style={draw, rectangle, minimum width=1.2cm, minimum height=1cm, align=center, rounded corners=5pt},
  ->, >={open triangle 60}
]
\node (A) {Tanzania};
\node (B) [right=of A] {Uganda};
\node (C) [above right=of B] {South Sudan};
\node (D) [below right=of C] {Sudan};
%\node (D) [right=of $(B)!0.5!(C)$, xshift=1cm] {Turkmenistan};
\node (E) [right=of D] {Egypt};
\node (F) [below right=of B, xshift=0.4cm] {Ethiopia};
\draw (A) -- (B);
\draw (B) -- (C);
\draw (C) -- (D);
\draw (D) -- (E);
\draw (F) -- (D);
\end{tikzpicture}
\caption{Flow structure of the Nile River.\label{graphNile}}
\end{figure}

We use AQUASTAT to estimate the inflows and actual water allocations for the six countries analyzed.\footnote{AQUASTAT, developed by the Food and Agriculture Organization (FAO), is a global information system focused on water resources and agricultural water management. It compiles, analyzes, and provides free access to over 180 variables and indicators by country. Free access to data is provided at \url{https://www.fao.org/aquastat/en/databases/}} 

As for inflows, Lake Victoria contributes a substantial 33 km$^3$/year to the Nile River. Its surface is divided among three countries: Tanzania (51\%), Uganda (43\%) and Kenya (6\%). Tanzania's inflow accounts for 51\% of Lake Victoria's contribution to the Nile River. Uganda's inflow is a combination of its portion of Lake Victoria's contribution (based on surface area) and water from Lake Albert, estimated at 2 km$^3$/year. In South Sudan, the Nile River receives water from two major tributaries: the Bahr al Ghazal River, which flows almost entirely within the country and discharges 1.5 km$^3$/year, and the Sobat River, formed by the confluence of the Pibor and Baro Rivers, contributing 3.1 and 13 km$^3$/year respectively. Together, South Sudan's inflow totals 17.6 km$^3$/year. Ethiopia's contribution comes from the Blue Nile, originating at Lake Tana, with an estimated inflow of 52.6 km$^3$/year. The next country in the course, Sudan, receives an additional 0.7 km$^3$/year from the Atbara River. Finally, according to AQUASTAT, Egypt relies entirely on upstream inflows, as it does not contribute water to the Nile River. The resulting amounts described in this paragraph for the six countries are listed in the second column of Table \ref{table1} ($e$). 

\begin{figure}[h]
  \centering
  \includegraphics[scale=0.7]{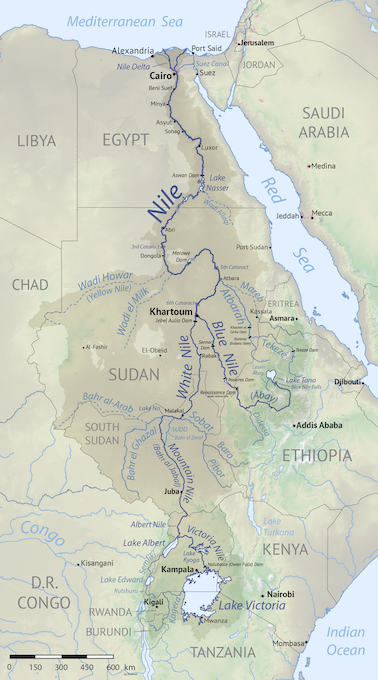}
  \caption{Nile River Basin \label{fig:NileBasin}}
\end{figure}

To estimate the actual water allocation among the countries along the Nile River, we have considered as a proxy the so-called \emph{total freshwater withdrawal} variable in AQUASTAT. The withdrawals from Tanzania, Uganda, South Sudan, Ethiopia, Sudan, and Egypt amount to 5.18 km$^3$/year, 0.64 km$^3$/year, 0.66 km$^3$/year, 10.55 km$^3$/year, 26.93 km$^3$/year and 77.7 km$^3$/year, respectively. The cumulative allocation of these withdrawals totals 121.66km$^3$/year, which is slightly above the overall inflow of 103.9 km$^3$/year. To be consistent with our theoretical model, we have proportionally adjusted the withdrawals so that the total sum is also 103.9 km$^3$/year. These adjustments are listed into the third column of Table \ref{table1} ($z$).

\begin{table}[H]
    \centering
    \begin{tabular}{lrrrrrrrrr}
    \toprule
    & & & \multicolumn{7}{c}{Water rights} \\
    \cmidrule(lr){4-10}
    Countries & $e$ & $z$ & $R^{FT}$ & $R^{\gamma=\frac{1}{4}}$ & $R^{\gamma=\frac{1}{2}}$ & $R^{\gamma=\frac{3}{4}}$ & $R^{NT}$ & $R^{S}$ & $R^{g^*}$  \\ 
    \midrule 
    Tanzania    & 16.8 & 4.42  & 0     & 4.2   & 8.4   & 12.6  & 16.8 & 3.36  & 4.42 \\
    Uganda      & 16.2 & 0.55  & 0     & 7.2   & 12.3  & 15.3  & 16.2 & 7.41  & 0.55 \\
    South Sudan & 17.6 & 0.56  & 0     & 9.8   & 14.95 & 17.02 & 17.6 & 13.28 & 0.56 \\
    Ethiopia    & 52.6 & 9.01  & 0     & 13.15 & 26.3  & 39.45 & 52.6 & 11.53 & 9.01 \\
    Sudan       & 0.7  & 22    & 0     & 17.38 & 20.97 & 14.64 & 0.7  & 31.16 & 22\\
    Egypt       & 0    & 66.35 & 103.9 & 52.16 & 20.97 & 4.88  & 0    & 31.16 & 66.35\\
    \bottomrule
    \end{tabular} 
    \caption{Inflows (e), actual allocations (z) and water-use rights in the Nile River ($g^*=(0.26,0.02,0.01,0.17,0.26)$).\label{table1}}
\end{table}

Table \ref{table1} yields different allocations coming from the rules characterized in our analysis.\footnote{Note that, in order to present these results, we have adjusted the rules from our analysis accordingly to account for the actual structure of the Nile River, which is not entirely linear, due to the location of Ethiopia.} 

Its last column actually provides the allocation that the multi-parameter geometric rule $R^{g^*}$ yields, when $g^*=(0.26,0.02,0.01,0.17,0.26)$. This happens to coincide with the actual water allocation ($z$). That is, our family of multi-parameter geometric rules rationalizes the actual water allocation by means of imposing upstream countries they retain rather small portions of their inflows (particularly, Uganda and South Sudan). 

Alternatively, we also provide the allocations some single-parameter geometric rules yield (when all countries are imposed the same portion $\gamma$ of their inflows, with $\gamma\in\{0, \frac{1}{4}, \frac{1}{2}, \frac{1}{4}, 1\}$), as well as the serial rule. With the obvious exception of the no-transfer rule, we can see that these allocations generally favor upstream countries (especially, Uganda and South Sudan) with respect to the actual water allocation.

\section{Discussion}
Our paper can be considered as a new instance of the literature dealing with river sharing, initiated by \cite{Ambec02} and followed by \cite{Ambec2008}, \cite{Ansink12}, \cite{Brink2012}, \cite{Gudmundsson2019}, and \cite{Oeztuerk2020}, among others.\footnote{\cite{beal2013} is a survey of early contributions within this literature.} As in \cite{Mart2025}, we depart from the existing literature to study the adjudication of water rights, rather than the end-state allocation of water. That is, we assume that end-state allocations of water arise after a two-stage process in which rights are assigned first, based only on objective information about agents, and an ensuing market to trade those rights occurs afterwards. We only concentrate on the first part of the process. In doing so, we do not consider utility functions to transform water into welfare, %These utility functions are, at the very least, difficult to elicit, 
which renders our model less demanding from an informational viewpoint. \cite{Ansink12} also analyze river sharing problems without utility functions. But, in their setting, agents also have claims (and not only inflows). Allocations can violate the feasibility constraints we impose in our analysis (which might be interpreted as allowing that water could be transferred upwards via external methods), but claims constitute upper bounds for allocations. Thus, the rules in their setting largely differ from the ones we consider here. 

A similar model is also used to analyze the problem of cleaning a polluted river, where the input is the cost associated to each agent located along the river \citep[e.g.,][]{Ni2007, Dong2012, AlcaldeUnzu2015, Brink2018}. 
%\footnote{See also \cite{Dong2012, AlcaldeUnzu2015, AlcaldeUnzu2020, Li2022} among others.}initiated by \cite{Ni2007},  , Li2022 Brink2018,
%are \emph{non-wasteful} and \emph{feasible}, i.e.,  $x=\left(x_1, \ldots, x_n\right) \in \mathbb{R}^n_+$ is such that $\sum_{i=1}^n x_i=\sum_{i=1}^n e_i$ and $\sum_{i=1}^k x_i\le\sum_{i=1}^k e_i$, for each $k=1,\dots, n-1$. 
In that setting, only non-wastefulness is imposed from the outset. But no other feasibility constraints are imposed. That is why rules used in that setting might not be well defined in ours. More recently, \cite{yang2025pollute} have studied the so-called \textit{river pollution claims problem}, where the aim is to distribute a budget of emissions permits among agents located along a river. In this case, the 
input is each agent's claim (reflecting population, emission history, and business-as-usual emissions), as well as a budget that is lower (or equal) than the aggregate claim. %For environmental reasons, the specific location along the river where pollutants are emitted is an important concern (the more upstream the location is the higher the damage of polluting the river).
\cite{martinez5193807fair} propose in that setting rules that adjust the geometric rules studied in this paper.

Finally, our model is also similar to the problem of revenue sharing in hierarchies \citep[e.g.,][]{Hougaard17}, in which the issue is to distribute the proceeds generated by agents organized on a hierarchy. In the benchmark case in which the hierarchy is a line, agents are thus just characterized by the location in the line and the revenue they bring to the hierarchy. Transfer rules determine a redistribution of the overall revenue along the line, with the proviso that revenues can only be transferred upwards in the hierarchy. As such, the model can simply be interpreted as the mirror image of ours (in which the bottom agent of the river plays the role of the agent at the top of the hierarchy) and our feasibility constraints for allocations are satisfied by the definition of (upward) transfer rules in the hierarchy. \cite{Hougaard17} characterize the family of geometric transfer rules in their setting, making use of suitable axioms therein that would not be naturally translated into our setting (with the exception of scale invariance, that is considered in both settings).\footnote{Intermediate geometric rules for revenue sharing in hierarchies are reminiscent of popular incentive mechanisms for social mobilization or multi-level marketing \citep[e.g.,][]{Pickard2011}.}   

We have also provided an illustration of our analysis to the case of the Nile River. We have identified our multi-parameter geometric rule rationalizing the actual allocation of water therein. This rule involves that upstream countries retain rather small portions of their inflows (particularly, Uganda and South Sudan). Alternative allocations coming from other rules in our analysis typically yield larger portions for those countries. %The so-called Nile Basin Initiative is an intergovernmental partnership, which was launched in February 1999 by the water ministers of the countries that share the river, to provide a forum for consultation and coordination among the Basin States.\footnote{https://nilebasin.org/} %for the sustainable management and development of the shared Nile Basin water and related resources for win-win benefits. 
%They work ``to achieve sustainable socio-economic development through the equitable utilization of, and benefit from, the common Nile basin water resources". We modestly believe our results might provide interesting insights for potential negotiations among the involved parties within the framework of this initiative. 
%with the aim of ``seeking to develop the river in a cooperative manner, share substantial socioeconomic benefits, and promote regional peace and security" and to ``provide an institutional mechanism, a shared vision, and a set of agreed policy guidelines to provide a basinwide framework for cooperative action."

\section{Appendix. Proofs of the results}
We first provide a lemma, which is interesting on its own, and will pave the way for the proofs of our theorems stated above. 

\begin{lemma}\label{UI+PII}
If a rule $R$ satisfies upstream invariance and partial-implementation invariance, then there exist $n-1$ functions $\alpha_1,\ldots,\alpha_{n-1}: \mathbb{R}_+ \longrightarrow \mathbb{R}_+$ such that, $\alpha_i(r) \le r$ for each $i \in \{1,\ldots,n-1\}$ and
$$
R_i(e)=R^{(\alpha_1,\ldots,\alpha_{n-1})}_i(e) = 
\begin{cases}
    \alpha_1(e_1) & \text{if } i=1 \\[0.1cm]
    \alpha_i\left( e_i + \sum_{k=1}^{i-1} \left( e_k-R^{(\alpha_1,\ldots,\alpha_{n-1})}_k(e) \right) \right)  & \text{if } i \in \{2,\ldots,n-1\} \\[0.1cm]
    e_n + \sum_{k=1}^{n-1} \left( e_k-R^{(\alpha_1,\ldots,\alpha_{n-1})}_k(e) \right)  & \text{if } i=n. 
\end{cases}
$$
\end{lemma}
\begin{proof}
%We focus on the non-trivial implication. 
Let $e \in \mathcal{D}$ and $R$ be a rule satisfying \textit{upstream invariance} and \textit{partial-implementation invariance}. For each $r\in\mathbb{R}_+$, let %us define each function $\alpha_i$ as follows
$$
\alpha_i(r) = R_i(\overbrace{0, \ldots, 0}^{i-1}, r, \overbrace{0, \ldots, 0}^{n-i}).
$$
Notice that, by definition of rule, $\alpha_i(r) \le r$. We proceed by induction on the number of non-null entries in $e$. By definition, $R(0,\ldots,0)=(0,\ldots,0)=R^{(\alpha_1,\ldots,\alpha_{n-1})}(0,\ldots,0)$. Let us define $e^n=(0,\ldots,0,e_n) \in \mathcal{D}$. Again, by definition, 
$$
R(e^n)=(0,\ldots,0,e_n)=R^{(\alpha_1,\ldots,\alpha_{n-1})}(e^n).
$$
Now, let $e^{n-1}=(0,\ldots,0,e_{n-1},e_n) \in \mathcal{D}$. By \emph{upstream invariance}, $R_{n-1}(e^{n-1}) = R_{n-1}(0,\ldots,0,e_{n-1},0) = \alpha_{n-1}(e_{n-1}) = R^{(\alpha_1,\ldots,\alpha_{n-1})}_{n-1}(e^{n-1})$. Therefore,
$$
R(e^{n-1})=(0,\ldots,0, \alpha(e_{n-1}), e_n + (e_{n-1}-\alpha(e_{n-1})))=R^{(\alpha_1,\ldots,\alpha_{n-1})}(e^{n-1}).
$$
Assume now that the claim holds for $e^i=(0,\ldots,0,e_i,e_{i+1},\ldots,e_n)$. We show that it is also true for $e^{i-1}=(0,\ldots,0,e_{i-1},e_i,\ldots,e_n)$. It is obvious that, for each $k<i-1$, $R_k(e^{i-1})=0=R^{(\alpha_1,\ldots,\alpha_{n-1})}_k(e^{i-1})$. Now, by \emph{upstream invariance} (as argued above), $R_{i-1}(e^{i-1}) = \alpha_{i-1}(e_{i-1}) = R^{(\alpha_1,\ldots,\alpha_{n-1})}_{i-1}(e^{i-1})$. Finally, by \emph{partial implementation invariance},
$$
R_k(e^{i-1}) = R_k \left( 0,\ldots,0, e_i+ \sum_{j=1}^{i-1} (e_j-R_j(e^{i-1})),e_{i+1},e_n \right). 
$$
Then, applying the induction hypothesis to the problem with one additional null inflow $\left( 0,\ldots,0, e_i+ \sum_{j=1}^{i-1} (e_j-R_j(e^{i-1})),e_{i+1},e_n \right)$, we obtain that
$$
R_k(e^{i-1})= R^{(\alpha_1,\ldots,\alpha_{n-1})}_k \left( 0,\ldots,0, e_i+ \sum_{j=1}^{i-1} (e_j-R_j(e^{i-1})),e_{i+1},e_n \right).
$$
We have already shown that, for each $j \leq i-1$, $R_j(e^{i-1})=R^{(\alpha_1,\ldots,\alpha_{n-1})}_j(e^{i-1})$. Thus, $\sum_{j=1}^{i-1} (e_j-R_j(e^{i-1})) = \sum_{j=1}^{i-1} (e_j-R^{(\alpha_1,\ldots,\alpha_{n-1})}_j(e^{i-1}))$ and, therefore,
$$
R_k(e^{i-1})= R^{(\alpha_1,\ldots,\alpha_{n-1})}_k \left( 0,\ldots,0, e_i+ \sum_{j=1}^{i-1} (e_j-R^{(\alpha_1,\ldots,\alpha_{n-1})}_j(e^{i-1})),e_{i+1},e_n \right) = R^{(\alpha_1,\ldots,\alpha_{n-1})}_k(e^{i-1}),
$$
as desired. 
\end{proof}

\subsection*{Proof of Theorem \ref{thm_generalized_geom}}

\begin{proof}
We prove first the straightforward implications of the statement. Let $\alpha=(\alpha_1,\ldots,\alpha_{n-1},\alpha_n) \in [0,1]^{n-1} \times \{1\}$ and $R^\alpha$ be the corresponding multi-parameter geometric rule. 
\begin{itemize}
    \item[(a)] \emph{Partial-implementation invariance}. Let $e \in \mathcal{D}$ and $i\in N$. We define $\overline{e}^i \in \mathcal{D}$ as
    $$
    \overline{e}^i = \left( 0_{U(i) \backslash \{i\}} , e_i+ \sum_{k=1}^{i-1} (e_k-R^\alpha_k(e)), e_{D(i) \backslash \{i\}} \right).
    $$
    We must check that $R^\alpha_j(\overline{e}^i)=R^\alpha_j(e)$ for each $j \in \{i,\ldots,n-1\}$. If $j=i$, by definition of multi-parameter geometric rules, we have that
    $$
    R^\alpha_i(\overline{e}^i)= \alpha_i \left( \overline{e}^i_i+ \sum_{k=1}^{i-1} (\overline{e}^i_k-R^\alpha_k(\overline{e}^i)) \right) = \alpha_i \overline{e}^i_i = \alpha_i \left( e_i + \sum_{k=1}^{i-1} (e_k-R^\alpha_k(e)) \right) = R^\alpha_i(e).
    $$
    If $j=i+1$, 
    \begin{align*}
        R^\alpha_{i+1}(\overline{e}^i) &= \alpha_{i+1} \left( \overline{e}^i_{i+1}+ \sum_{k=1}^{i} (\overline{e}^i_k-R^\alpha_k(\overline{e}^i)) \right) \\
        &= \alpha_{i+1} \left( \overline{e}^i_{i+1} + (\overline{e}^i_i-R^\alpha_i(\overline{e}^i))  \right)  \\
        &= \alpha_{i+1} \left( e_{i+1} + e_i+ \sum_{k=1}^{i-1} (e_k-R^\alpha_k(e)) -R^\alpha_i(e)  \right) \\
        &= \alpha_{i+1} \left( e_{i+1} + \sum_{k=1}^{i} (e_k-R^\alpha_k(e)) \right)  \\
        &= R^\alpha_{i+1}(e).
    \end{align*}
    Assume, by induction, that it holds for agent $j$, we now prove that it is also true for agent $j+1 \in \{i,\ldots,n-1\}$.  
    \begin{align*}
        R^\alpha_{j+1}(\overline{e}^i) &= \alpha_{j+1} \left( \overline{e}^i_{j+1}+ \sum_{k=1}^{j} (\overline{e}^i_k-R^\alpha_k(\overline{e}^i)) \right) \\
        &= \alpha_{j+1} \left( \overline{e}^i_{j+1} + (\overline{e}^i_i-R^\alpha_i(\overline{e}^i)) + \sum_{k=i+1}^{j} (\overline{e}^i_k-R^\alpha_k(\overline{e}^i)) \right)  \\
        &= \alpha_{j+1} \left( e_{j+1} + e_i+ \sum_{k=1}^{i-1} (e_k-R^\alpha_k(e)) -R^\alpha_i(e) + \sum_{k=i+1}^{j} (e_k-R^\alpha_k(e)) \right) \\
        &= \alpha_{j+1} \left( e_{j+1} + \sum_{k=1}^{j} (e_k-R^\alpha_k(e)) \right)  \\
        &= R^\alpha_{j+1}(e).
    \end{align*}
    Notice that, by definition of rule, we also have that $R^\alpha_n(\overline{e}^i)=R^\alpha_n(e)$. Therefore, we conclude that $R^\alpha_j(\overline{e}^i)=R^\alpha_j(e)$, for each $j \in \{i,\ldots,n\}$, and the property is satisfied.
    \item[(b)] \emph{Upstream invariance}. Let $e,e' \in \mathcal{D}$, such that $e_i<e'_i$ for some $i \in N$, and $e_j = e'_j$ for all $j \in N \backslash \{i\}$, for each $k<i$ we have that
    $$
    R^\alpha_k(e') = \alpha_k \left( e'_k + \sum_{j=1}^{k-1} \prod_{l=j}^{k-1} (1-\alpha_l)e'_j \right) = \alpha_k \left( e_k + \sum_{j=1}^{k-1} \prod_{l=j}^{k-1} (1-\alpha_l)e_j \right) = R^\alpha_k(e),
    $$
    as desired.
    \item[(c)] \emph{Scale invariance}. Let $\rho \in \mathbb{R}_+$. For each $e \in \mathcal{D}$ and each $i \in \{1,\ldots,n-1\}$,
    $$
    R^\alpha_i(\rho e) = \alpha_i \left( \rho e_i + \sum_{k=1}^{i-1} \prod_{j=k}^{i-1} (1-\alpha_j) \rho e_k \right) =  \rho \left[ \alpha_i \left( e_i + \sum_{k=1}^{i-1} \prod_{j=k}^{i-1} (1-\alpha_j)e_k \right) \right] = \rho R^\alpha_i(e),
    $$
    and
    $$
    R^\alpha_n(\rho e) = \rho e_n + \sum_{k=1}^{n-1} \prod_{j=k}^{n-1} (1-\alpha_j) \rho e_k = \rho \left[ e_n + \sum_{k=1}^{n-1} \prod_{j=k}^{n-1} (1-\alpha_j)e_k \right] = \rho  R^\alpha_n(e),
    $$
    as desired.
\end{itemize} 
We now focus on the converse implication. Let $e \in \mathcal{D}$ and $R$ be a rule satisfying \textit{partial-implementation invariance}, \textit{upstream invariance} and \textit{scale invariance}. By Lemma \ref{UI+PII}, there exist $n-1$ functions $\alpha_1,\ldots,\alpha_{n-1}: \mathbb{R}_+ \longrightarrow \mathbb{R}_+$ such that $\alpha_i(r) \le r$ for each $i \in \{1,\ldots,n-1\}$, and
$$
R_i(e)=R^{(\alpha_1,\ldots,\alpha_{n-1})}_i(e) = 
\begin{cases}
    \alpha_1(e_1) & \text{if } i=1 \\[0.1cm]
    \alpha_i\left( e_i + \sum_{k=1}^{i-1} \left( e_k-R^{(\alpha_1,\ldots,\alpha_{n-1})}_k(e) \right) \right)  & \text{if } i \in \{2,\ldots,n-1\} \\[0.1cm]
    e_n + \sum_{k=1}^{n-1} \left( e_k-R^{(\alpha_1,\ldots,\alpha_{n-1})}_k(e) \right)  & \text{if } i=n. 
\end{cases}
$$
Furthermore, 
$$
\alpha_i(r) = R_i(\overbrace{0, \ldots, 0}^{i-1}, r, \overbrace{0, \ldots, 0}^{n-i}).
$$
Now, by \textit{scale invariance}, 
$$
R_i(\overbrace{0, \ldots, 0}^{i-1}, r, \overbrace{0, \ldots, 0}^{n-i})=r R_i(\overbrace{0, \ldots, 0}^{i-1}, 1, \overbrace{0, \ldots, 0}^{n-i}).
$$
Thus, $\alpha_i(r)= r \alpha_i(1)$. Let $\alpha_i=\alpha_i(1)\in[0,1]$. Then  
%from where it follows that $G$ is a linear function. That is, 
it follows that, for each $i \in \{1,\ldots,n-1\}$, there exists $\alpha_i\in[0,1]$ such that $\alpha_i(r)=\alpha_i r$. Altogether, we have that $R$ is a multi-parameter geometric rule.    
\end{proof}

\subsection*{Proof of Theorem \ref{thm_geom}}

\begin{proof}
%We focus on the non-trivial implication. 
By Theorem \ref{thm_generalized_geom}, it follows that each single-parameter geometric satisfies \textit{upstream invariance}, \textit{partial-implementation invariance} and \textit{scale invariance}. It is straightforward to show that they also satisfy \textit{equal sources}. Conversely, let $R$ be a rule satisfying the four axioms and let $e \in \mathcal{D}$. %\textit{upstream invariance}, \textit{partial-implementation invariance}, \textit{scale invariance} and \textit{equal sources}. 
By Theorem \ref{thm_generalized_geom}, there exists $(\alpha_1,\ldots,\alpha_{n-1})\in [0,1]^{n-1}$ such that 
$$
R_i(e)=R^{(\alpha_1,\ldots,\alpha_{n-1})}_i(e) = 
\begin{cases}
    \alpha_1e_1 & \text{if } i=1 \\[0.1cm]
    \alpha_i\left( e_i + \sum_{k=1}^{i-1} \left( e_k-R^{(\alpha_1,\ldots,\alpha_{n-1})}_k(e) \right) \right)  & \text{if } i \in \{2,\ldots,n-1\} \\[0.1cm]
    e_n + \sum_{k=1}^{n-1} \left( e_k-R^{(\alpha_1,\ldots,\alpha_{n-1})}_k(e) \right)  & \text{if } i=n. 
\end{cases}
$$
Furthermore, 
$$
\alpha_i = R_i(\overbrace{0, \ldots, 0}^{i-1}, 1, \overbrace{0, \ldots, 0}^{n-i}).
$$
Now, by \textit{equal sources}, $\alpha_i=\alpha_j$ for each pair $i,j\in \{2,\ldots,n-1\}$. Altogether, we have that $R$ is a single-parameter geometric.    
\end{proof}

\subsection*{Proof of Theorem \ref{thm_shapley}}
%\subsection*{Proof of Theorem \ref{thm_shapley}}
\begin{proof}
Note first that the serial rule coincides with the multi-parameter geometric rule when $\alpha_i = \frac{1}{n-i+1}$ for each $i \in \{1,\ldots,n-1\}$. %Let $e \in \mathcal{D}$. If $i=1$, we have that
Thus, it satisfies partial-implementation invariance and upstream invariance (Theorem \ref{thm_generalized_geom}). %and scale invariance 
As for neutrality, let $e \in \mathcal{D}$, such that $e_i>0$ for some $i \in \{1,\dots, n-1\}$, and $e_j=0$ for all $j \in N \backslash \{i\}$. Then,
$$
R^{S}_i(e)=\alpha_i e_i = \frac{1}{n-i+1} e_i= \frac{1}{n-i} \left[ e_i - \frac{1}{n-i+1}e_i \right]= \frac{1}{n-i} \sum_{k=i+1}^{n} R^{S}_k(e),
$$  
as desired. 

Conversely, let $R$ be a rule that satisfies the axioms in the statement. By Lemma \ref{UI+PII}, we know that, for each $e \in \mathcal{D}$,
$$
R_i(e) = 
\begin{cases}
    \alpha_1(e_1) & \text{if } i=1 \\[0.1cm]
    \alpha_i\left( e_i + \sum_{k=1}^{i-1} \left( e_k-R_k(e) \right) \right)  & \text{if } i \in \{2,\ldots,n-1\} \\[0.1cm]
    e_n + \sum_{k=1}^{n-1} \left( e_k-R_k(e) \right)  & \text{if } i=n, 
\end{cases},
$$
where, for each $i \in \{1,\ldots,n-1\}$, $\alpha_i(e_i)= R_i(\overbrace{0, \ldots, 0}^{i-1}, e_i, \overbrace{0, \ldots, 0}^{n-i})$. 

Now, by \emph{neutrality},
$$
\alpha_i(e_i)=\frac{1}{n-i}(e_i-\alpha_i(e_i)).
$$
Equivalently, $\alpha_i(e_i)=\frac{e_i}{n-i+1}$. Let $\alpha_i=\frac{1}{n-i+1}$. Then, $\alpha_i(e_i)=\alpha_ie_i$, and $\alpha_i \in [0,1]$ because $0 \le \alpha_i(e_i) \le e_i$. That is, $R$ is actually a multi-parameter geometric rule in which $\alpha_i=\frac{1}{n-i+1}$ for each $i \in \{1,\ldots,n-1\}$, and $\alpha_n=1$, which concludes the proof. %As $R^{S}$ is the $R^\alpha$ rule with these specific $\alpha_i$'s, the proof is complete.
\end{proof}

\subsection*{Proof of Theorem \ref{thm_add_geom}}
\begin{proof}
We start by showing that each member of the $\beta$-family satisfies the axioms in the statement. By definition, each $\beta$-rule is a multi-parameter geometric rule. Thus, it saltisfies scale invariance, upstream invariance and partial-implementation invariance. Besides, each $\beta$-rule is also a $R^\delta$ rule. Indeed, for each $\beta_k \in [0,1)$, let us define $\delta(\beta_k)=\delta^k=(\delta^k_1, \ldots,\delta^k_n) \in [0,1]^{n-1} \times \{1\}$ as follows:
$$
\delta^k_i = 
\begin{cases}
1 & \text{if } i<k \\
\beta_k & \text{if } i=k \\
\frac{1}{n-i+1} & \text{if } k<i<n.
\end{cases}
$$
Then, let $e \in \mathcal{D}$ and $i \in N$. We consider three cases.
\begin{itemize}
\item[(i)] $i<k$. By definition of $\delta^k$,
$$
R^{\delta_k}_i(e) = \delta^k_i e_i + \sum_{j<i} \frac{1-\delta^k_j}{n-j} e_j = e_i.
$$
By definition of $R^{\beta_k}$, as $i<k$, $\alpha_r=1$ for all $r \in \{1,\ldots,i\}$, and then
$$
R^{\beta_k}_i(e) = \alpha_i \left[ e_i + \sum_{j<i} \prod_{j \le r <i} (1-\alpha_r) e_j \right] = e_i.
$$
Therefore, when $i<k$, $R^{\delta_k}_i(e) = R^{\beta_k}_i(e)$.
\item[(ii)] $i=k$. We compute both $R^{\delta_k}_k(e)$ and $R^{\beta_k}_k(e)$ to check that they coincide:
$$
R^{\delta_k}_k(e) = \delta^k_k e_k + \sum_{j<k} \frac{1-\delta^k_j}{n-j} e_j = \delta^k_k e_k = \beta_k e_k,
$$
and
$$
R^{\beta_k}_k(e) = \alpha_k \left[ e_k + \sum_{j<k} \prod_{j \le r <k} (1-\alpha_r) e_j \right] = \beta_k e_k.
$$
\item[(iii)] $k<i \le n$. Let us start by noting that
\begin{equation}\label{eq_prod}
\prod_{j \le r <i} (1-\alpha_r) = \prod_{j \le r <i} \frac{n-r}{n-r+1} = \frac{n-i+1}{n-j+1}.
\end{equation}
By definition of $R^{\beta_k}$, as $\alpha_r=1$ for all $r <k$, we have that
\begin{align*}
R^{\beta_k}_i(e) &= \alpha_i \left[ e_i + \sum_{j<k} \prod_{j \le r <i} (1-\alpha_r) e_j + \prod_{k \le r <i} (1-\alpha_r) e_k + \sum_{k<j<i} \prod_{j \le r <i} (1-\alpha_r) e_j \right] \\
&= \alpha_i \left[ e_i + \prod_{k \le r <i} (1-\alpha_r) e_ k+ \sum_{k<j<i} \prod_{j \le r <i} (1-\alpha_r) e_j \right].
\end{align*}
Using (\ref{eq_prod}), and the definition of each $\alpha_i$ in the $\beta_k$-subfamily, the previous expression is simplified to:
\begin{align*}
R^{\beta_k}_i(e) &= \frac{1}{n-i+1} e_i + \frac{1}{n-i+1} \frac{n-i+1}{n-(k+1)+1} (1-\beta_k) e_k + \frac{1}{n-i+1} \sum_{k<j<i} \frac{n-i+1}{n-j+1} e_j \\
&= \frac{1}{n-i+1} e_i + \frac{(1-\beta_k)}{n-k} e_k + \sum_{k<j<i} \frac{1}{n-j+1} e_j.
\end{align*}
On the other hand, by definition of $\delta^k$,
\begin{align*}
R^{\delta^k}_i(e) &= \delta^k_i e_i + \sum_{j<k} \frac{(1-\delta^k_j)}{n-j} e_j + \frac{(1-\delta^k_k)}{n-k} e_k + \sum_{k<j<i} \frac{(1-\delta^k_j)}{n-j} e_j\\
&= \frac{1}{n-i+1} e_i + \frac{(1-\beta_k)}{n-k} e_k + \sum_{k<j<i}  \frac{1}{n-j+1} e_j.\\
\end{align*}
Therefore, when $k<i \le n$, $R^{\delta_k}_i(e) = R^{\beta_k}_i(e)$.
\end{itemize}
Altogether, 
$$
R^{\delta^k}_i(e) = \delta^k_i e_i + \sum_{j<i} \frac{(1-\delta^k_j)e_j}{n-j} =
\begin{cases}
1 & \text{if } i<k \\
\beta_k & \text{if } i=k \\
\frac{1}{n-i+1} & \text{if } k<i<n,
\end{cases}
$$
and
$$
R^{\beta_k}_i(e) = \alpha_i \left( e_i + \sum_{k=1}^{i-1} \prod_{j=k}^{i-1} (1-\alpha_j)e_k \right).
$$
Thus, we have shown that each $\beta$-rule belongs to the $R^\delta$ family, which guarantees (by Proposition \ref{thm_add}) that each $\beta$-rule also satisfies downstream impartiality. 

Now, we focus on the converse implication. That is, let $R$ be a rule satisfying the four axioms in the statement. By Theorem \ref{thm_generalized_geom} and Proposition \ref{thm_add}, $R$ belongs to the intersection of families $R^\alpha$ and $R^\delta$. Let $e^i \in \mathcal{D}$ be such that $e^i_i=1$ and $e^i_j=0$ for all $j \in N \backslash \{i\}$. Then, $R^\alpha(e^i) = R^\delta(e^i)$ for all $i \in \{1,\ldots,n-1\}$. Thus, by definition,
$$
\alpha_1=R^\alpha_1(e^1) = R^\delta_1 (e^1) = \delta_1.
$$
We distinguish two cases.
\begin{itemize}
\item[(1.a)] $\alpha_1=\delta_1\neq 1$. As $R^\alpha_2(e^1) = R^\delta_2(e^1)$, we have that
$$
\alpha_2 (1-\alpha_1) = \frac{1-\delta_1}{n-1} \; \equiv \; \alpha_2=\frac{1}{n-1}.
$$
Analogously, $R^\alpha_3(e^1) = R^\delta_3(e^1)$ implies that
$$
\alpha_3 (1-\alpha_2) (1-\alpha_1) = \frac{1-\delta_1}{n-1} \; \equiv \; \alpha_3=\frac{1}{n-2}.
$$
Now, we proceed by induction. Assuming that $\alpha_{j}=\frac{1}{n-j+1}$ for each $1<j<k$, we prove that $\alpha_k=\frac{1}{n-k+1}$. Indeed, as $R^\alpha_k(e^1) = R^\delta_k(e^1)$, we have that
$$
 \alpha_k(1-\alpha_{k-1})\cdots (1-\alpha_1) = \frac{1-\delta_1}{n-1} \; \equiv \; \alpha_k = \frac{1}{n-k+1}.
$$
Therefore, in this case, $\alpha_1 \in [0,1)$ and $\alpha_i=\frac{1}{n-i+1}$ for each $i \in \{2,\ldots,n-1\}$.
\item[(1.b)] $\alpha_1=\delta_1 = 1$. Let us consider the problem $e^2 \in \mathcal{D}$. As $R^\alpha_2(e^2) = R^\delta_2(e^2)$, then $\alpha_2=\delta_2$. We distinguish two cases:
\begin{itemize}
\item[(2.a)] $\alpha_2=\delta_2 \neq 1$. As $R^\alpha_3(e^2) = R^\delta_3(e^2)$, we have that
$$
\alpha_3 (1-\alpha_2) = \frac{1-\delta_2}{n-2} \; \equiv \; \alpha_3=\frac{1}{n-2}.
$$
As in case (1.a), we proceed by induction. Assuming that $\alpha_{j}=\frac{1}{n-j+1}$ for each $2<j<k$, we prove that $\alpha_k=\frac{1}{n-k+1}$. Indeed, as $R^\alpha_k(e^2) = R^\delta_k(e^2)$, we have that
$$
 \alpha_k(1-\alpha_{k-1})\cdots (1-\alpha_2) = \frac{1-\delta_2}{n-1} \; \equiv \; \alpha_k = \frac{1}{n-k+1}.
$$
Therefore, in this case, $\alpha_1=1$, $\alpha_2 \in [0,1)$ and $\alpha_i=\frac{1}{n-i+1}$ for each $i \in \{3,\ldots,n-1\}$.
\item[(2.b)] $\alpha_2=\delta_2 = 1$. Let us consider the problem $e^3 \in \mathcal{D}$. As $R^\alpha_3(e^3) = R^\delta_3(e^3)$, then $\alpha_3=\delta_3$. Again, we have two cases:
\begin{itemize}
\item[(3.a)] $\alpha_3=\delta_3 \neq 1$. Arguing as in case (2.a) we conclude that $\alpha_1=\alpha_2=1$, $\alpha_3 \in [0,1)$ and $\alpha_i=\frac{1}{n-i+1}$ for each $i \in \{4,\ldots,n-1\}$.
\item[(3.b)] $\alpha_3=\delta_3 = 1$... %We can iteratively replicate the reasoning until the number of parameters $\alpha_k$ is exhausted to demonstrate that, for each $k \in \{4,\ldots,n-1\}$, such that $\alpha_j = \delta_j =1$ for all $j<k$, it holds that $\alpha_j=1$ for each $j<k$, $\alpha_k \in [0,1)$, and $\alpha_i=\frac{1}{n-i+1}$ for each $k<i<n$.
\end{itemize}
\end{itemize}
\end{itemize}
We can proceed iteratively to conclude the proof.
\end{proof}

\newpage

%\bibliography{../references.bib}
%\bibliographystyle{../mystyle3}
%\bibliography{JR_references.bib} % Para overleaf
%\bibliographystyle{mystyle3}

%\newpage
\end{document}